\documentclass[12pt]{amsart}

\usepackage[mathscr]{eucal}
\usepackage{amsmath,amssymb,amscd}%,latexsym
\usepackage{color}
\usepackage{epsfig}
\newtheorem{theorem}{Theorem}
\newtheorem{conjecture}{Conjecture}

\newtheorem{lemma}{Lemma}

\newtheorem{proposition}{Proposition}
\theoremstyle{definition}

\theoremstyle{plane}

\def \beq{ \begin{equation} }
\def \eeq{\end{equation}}

\begin{document}
\title{Saari's Homographic Conjecture of the Three-Body Problem}

\renewcommand\baselinestretch{}
\authors{Florin Diacu$^{*}$, Toshiaki Fujiwara$^{**}$, Ernesto P\'erez-Chavela$^{\dag}$, and Manuele
Santoprete$^{\dag\dag}$}

\address{\noindent $^{*}$ Department of Mathematics and Statistics,
University of Victoria, Victoria, BC, Canada}
\address{\noindent $^{**}$ College of Liberal Arts and Sciences,
Kitasato University, 1-15-1 Kitasato, Sagamihara, Kanagawa 228-8555, Japan}
\address{\noindent $^{\dag}$ Departamento de Matem\'aticas, UAM--Iztapalapa, A.P. 55--534,
09340 Iztapalapa, Mexico, D.F., Mexico}
\address{\noindent $^{\dag\dag}$ Department of Mathematics,
 Wilfrid Laurier University, Waterloo, ON, Canada}

\email{diacu@math.uvic.ca,
fujiwara@clas.kitasato-u.ac.jp, \\
epc@xanum.uam.mx, msantoprete@wlu.ca}



\begin{abstract}
Saari's homographic conjecture, which extends a classical statement proposed by Donald Saari in 1970, claims that solutions of the Newtonian $n$-body problem with constant configurational measure are homographic. In other words, if the mutual distances satisfy a certain relationship, the configuration of the particle system may change size and position but not shape. We  prove this conjecture for large sets of initial conditions in three-body problems given by homogeneous potentials, including the Newtonian one. Some of our results are true for $n\ge 3$.
\date{\large{\today}}
\end{abstract}

\maketitle

\markboth{F. Diacu, T. Fujiwara, E. P\'erez-Chavela, and M. Santoprete}{Saari's Homographic Conjecture}

%%%%%%%%%%%%%%%%%%%
%%%%%%%%%%%%%%%%%%%
%%%%%%%%%%%%%%%%%%%
\section{Introduction}
%%%%%%%%%%%%%%%%%%%
%%%%%%%%%%%%%%%%%%%
%%%%%%%%%%%%%%%%%%%
%Saari's conjecture is notoriously difficult. Donald Saari proposed it in
In 1970  Donald Saari proposed the following conjecture \cite{saari}: \emph{In the Newtonian $n$-body problem, if the moment of inertia, $I=\sum_{k=1}^nm_k|q_k|^2$, is constant, where $q_1, q_2,\dots, q_n$ represent the position vectors of the bodies of masses $m_1,\dots, m_n$, then the corresponding solution is a relative equilibrium}. In other words: Newtonian particle systems of constant moment of inertia rotate like rigid bodies.

A lot of energy has been spent to understand Saari's conjecture, but most of this work failed to achieve crucial results. An early attempt at a proof using variational techniques was, unfortunately, incorrect \cite{palmore1, palmore2}. More recently, the interest in this conjecture has grown considerably due to the discovery of the ``figure eight'' solution \cite{eight}, which---as numerical arguments show---has an approximately constant moment of inertia but is not a relative equilibrium.

Still, there have been a few successes in the struggle to understand Saari's conjecture. McCord proved that the conjecture is true for three bodies of equal masses \cite{mccord}. Llibre and Pi\~na gave an alternative proof of this case, but they never published it \cite{llibre}. Moeckel obtained a computer-assisted proof for the Newtonian three-body problem for any values of the masses \cite{moeckel, moeckel2005}. Diacu, P\'erez-Chavela, and Santoprete showed that the conjecture is true for any $n$ in the collinear case for potentials that depend only on the mutual distances between point masses \cite{collinear}. There have also been results, such as \cite{chen97, ffoLem, ffoI, roberts, santoprete, schmah}, which consider the conjecture in other contexts than the Newtonian one.

%Even more difficult than Saari's conjecture is Saari's homographic conjecture\footnote{Donald Saari proposed this extension of his original conjecture during an AMS-SIAM series of CMBS lectures he gave in June 2002 in Illinois.  }
%%%%%%%%%%
%%%%%%%%%%
A natural extension of the original Saari's conjecture, namely Saari's homographic conjecture, was proposed by Donald Saari  during an AMS-SIAM series of CMBS lectures he gave in June 2002 in Illinois. This extended conjecture was then emphasized  by Donald Saari during his talk at the Saarifest2005 and in his book \cite{saari-book}. In this extended conjecture  the role of the moment of inertia is played by the so called configurational measure. 

In the Newtonian $n$-body case  $UI^{1/2}$ defines the configurational measure of the particle system (also called scaled potential), where $U$ is the Newtonian potential. Saari's homographic conjecture claims that \emph{every solution of constant configurational measure is homographic} (see Section \ref{notations} for more details). In particular, if the moment of inertia is constant, it can be shown that the potential $U$ is constant, therefore the configurational measure is constant and every homographic solution with constant moment of inertia is a relative equilibrium. This reasoning shows why Saari's conjecture is a particular case of Saari's homographic conjecture. 

Saari's homographic conjecture covers new territory. While in Saari's conjecture collisions are excluded (because they lead to an unbounded potential, therefore to a non-constant moment of inertia) and the motion remains bounded (because the moment of inertia is constant), both collision and unbounded orbits may occur in Saari's homographic conjecture.

Moreover, Saari's homographic conjecture is important for astronomy. Homographic solutions are known to exist in the solar system. It is often said that the Sun, Jupiter, and the Trojan cluster of asteroids move like a relative equilibrium solution of the three-body problem, but this fact is true only in a first approximation. Jupiter orbits an ellipse around the Sun, although this ellipse is almost a circle. Therefore the motion of the system formed by the Sun, Jupiter, and the Trojan asteroids is better described as a homographic solution than as a relative equilibrium of the three-body problem.

In this paper we are primarily interested in Saari's homographic conjecture of three-body problems given by homogeneous potentials, but some of our results are also true for any number $n\ge3$ of bodies. The main results of this paper are given in Theorems 1 through 8. 

Theorem 1 shows that, for homogeneous potentials of order $-a$, with $a<2$, Saari's homographic conjecture is true for any total-collision solution of the planar or spatial $n$-body problem. Theorem 2 validates Saari's homographic conjecture for $0<a<2$ for any type of collision in the $n$-body case. Theorem 3 shows that Saari's homographic conjecture is correct in the rectilinear case for $0<a<2$. Theorem 4 shows Saari's homographic conjecture to be always valid in the collinear case. Theorem 5 proves that, for $0<a<2$, Saari's homographic conjecture is true in the three-body problem if the solutions stay away from the paths that make them scatter asymptotically towards rectilinear central configurations. Theorem 6 proves that Saari's homographic conjecture is correct in the Newtonian three-body problem with equal masses and non-negative energy. Theorem 7 shows that for any given initial configuration of three bodies, Saari's homographic conjecture is valid if the chosen angular momentum is large enough. Finally, Theorem 8 proves that if the angular momentum is chosen first, then Saari's homographic conjecture is true if the initial positions are taken close enough to an equilateral triangle of a certain size.   

The key tool for obtaining these results is provided by what we call Fujiwara coordinates, which were introduced by one of us. The motivation behind defining them is explained in Section \ref{secDecomposition}. 

The paper is structured as follows. Section \ref{notations} introduces some notations and definitions and gives a precise statement of Saari's homographic conjecture. Section \ref{secEvolutionOfI} deals with the evolution of the moment of inertia, which is then used in Section \ref{secCollisionOrbit} to prove Theorems 1, 2, and 3. Besides introducing the Fujiwara coordinates, Section \ref{secDecomposition} proves a lemma that characterises the homographic solutions relative to the new Fujiwara variables and proves Theorem 4. Section \ref{kineticSection} studies the evolution of the triangle's shape formed by three bodies and the connection of this shape with central configurations and homographic solutions. In Section \ref{candidate} we characterise the non-homographic solutions that would contradict Saari's homographic conjecture. This characterisation is used in the next sections towards showing that non-homographic solutions are impossible. Section \ref{condition} derives a condition that must be satisfied by the non-homographic candidate. This condition is crucial for proving the other theorems. Sections \ref{one} and \ref{two}, which prove Theorems 5 and 6, are separated by Section \ref{break}, whose goal is to study the analytic behaviour of the solutions near rectilinear central configurations. Finally, Section \ref{final} clarifies Theorems 7 and 8.

%%%%%%%%%%%%%%%%%%%
%%%%%%%%%%%%%%%%%%%
%%%%%%%%%%%%%%%%%%%
\section{Notations and definitions}
\label{notations}
%%%%%%%%%%%%%%%%%%%
%%%%%%%%%%%%%%%%%%%
%%%%%%%%%%%%%%%%%%%
In this paper we consider point-mass problems given 
by {\it homogeneous potential functions} of the form
\[
U = \frac{1}{a} \sum \frac{m_j m_k}{r_{jk}^a},
\]
where $a>0$ is a constant, $m_k$ are the masses, $r_{jk}=|q_j-q_k|$ define the mutual distances between bodies, and $q_k$ represent the coordinates of the point masses. (The terms {\it point mass}, {\it body}, and {\it particle} denote the same concept. Similarly, {\it solution} and {\it orbit} describe identical mathematical objects.) We are going to study the equations of motion given by the system of differential equations
\[
\begin{split}
&m_k \frac{dq_k}{dt}=p_k,\\
&\frac{dp_k}{dt}
	=\sum_{j\ne k} \frac{m_j m_k}{r_{jk}^{a+2}}(q_j-q_k)
	=g_k(q),
\end{split}
\]
where $q$ is a generic notation for the variables $q_k$. The {\it kinetic energy}, $T$, and the {\it total energy}, $H$, are defined as
\[
\begin{split}
T=\frac{1}{2}\sum \frac{|p_k|^2}{m_k}\ {\rm and}\
H=T-U,\ {\rm respectively.}
\end{split}
\]
Without loss of generality, we can fix the centre of mass of the particle system at the origin of the frame by taking
\[
\sum m_k q_k=0.
\]
We define the {\it moment of inertia},  $I$,  and the {\it angular momentum}, $C$, as
\[
I 
= \sum m_k |q_k|^2
=M^{-1}\sum m_j m_k |q_j-q_k|^2\ {\rm and}\ 
C = \sum q_k \wedge p_k,
\]
where
$M=\sum m_k$ is the {\it total mass} and $a\wedge b$ represents the {\it outer product} of the vectors $a$ and $b$. We use $a \cdot b$ for the {\it inner product}.

Notice that the moment of inertia, $I$, and the potential, $U$,
are homogeneous functions of degree $2$ and $-a$, respectively, so both functions are inversely proportional in the sense that one increases while the other one decreases. This observation leads to the notion of {\it configurational measure}, $UI^{a/2}$, which is a homogeneous function of degree zero.

A solution of our problem is called {\it homographic} if
the configuration of the particles remains similar with itself for all times. In other words, there exists a scalar $R=R(t)>0$
and an orthogonal matrix $\Omega=\Omega(t)$ such that, for all t,
$$q_k(t)=R(t)\Omega(t)q_k(t_0),$$
where $t_0$ is some fixed moment in time.

In particular, a solution is called homothetic if dilation/contraction occurs without rotation. This situation happens when $\Omega$ is the unit matrix. Similarly, a solution is called a relative equilibrium if rotation takes place without dilation/contraction. This scenario appears when $R=1$.

Our goal here is to shed light on the following conjecture. Most of our results will be restricted to the three-body case.

\begin{conjecture}[Saari's homographic]
\label{conjecture1}
If in the $n$-body problem given by a homogeneous potential of degree $-a$, with $a>0$, the configurational measure is constant, then the corresponding solution is homographic.
\end{conjecture}

%%%%%%%%%%%%%%%%%%%
%%%%%%%%%%%%%%%%%%%
%%%%%%%%%%%%%%%%%%%
\section{Evolution of the moment of inertia}
\label{secEvolutionOfI}
%%%%%%%%%%%%%%%%%%%
%%%%%%%%%%%%%%%%%%%
%%%%%%%%%%%%%%%%%%%
In this section, we consider the evolution of the moment of inertia
for constant configurational measure in the general planar, $\mathbb{R}^2$,  or spatial, $\mathbb{R}^3$,  $n$-body problem.
%%%%%%%%%
So, unless otherwise specified, we assume the configurational measure constant, and write
\[
\mu = U I^{a/2}.
\]
The Lagrange-Jacobi identity yields
\[
\frac{d^2 I}{dt^2}
=2\sum \frac{|p_k|^2}{m_k} - 2a U
=4H+2(2-a)U
=4H + 2(2-a)\mu I^{-a/2},
\]
where $H$ is the total energy. In the derivation of this identity we used the fact that
$
U
=\left(U I^{a/2}\right)I^{-a/2}
=\mu I^{-a/2}.
$
Integrating the Lagrange-Jacobi identity for $a\ne 2$,
we obtain
\[
\frac{1}{2}
	\left(
		\frac{dI}{dt}
	\right)^2
	+
	\left(
		-4HI - 4\mu I^{(2-a)/2}
	\right)
	= -2B,
\]
where $B$ is an integration constant. Let us write
\begin{equation}
\label{eqPhi}
\Phi(I)
=-4HI - 4\mu I^{(2-a)/2}.
\end{equation}
Then the evolution of the moment of inertia for $\mu=$ constant is described by the equation
\begin{equation}
\label{eqMotForI}
\frac{d^2 I}{dt^2}
= -\frac{\partial\Phi(I)}{\partial I},
\end{equation}
which leads to the first integral
\begin{equation}
\label{eqIntegralForI}
\frac{1}{2}
	\left(
		\frac{dI}{dt}
	\right)^2
	+
	\Phi(I)
	= -2B.
\end{equation}
Notice that the physical region of $\Phi(I)$ is $\Phi(I)\le 0$
because
\begin{equation}
\label{valueOfPhi}
\Phi(I)=-4HI-4UI=-4I(H+U)=-4IT.
\end{equation}
Moreover, the constant $B$ must be positive or zero. To see this, combine equations (\ref{eqIntegralForI}) and (\ref{valueOfPhi}) to obtain
\begin{equation}
\label{4ITMinus2B}
4IT-2B
=\frac{1}{2}\left( \frac{dI}{dt} \right)^2
=2 \left( \sum \sqrt{m_k}\;q_k \cdot {p_k\over\sqrt{m_k}} \right)^2.
\end{equation}
The application of the Cauchy-Schwartz inequality,
\begin{equation}
\label{CauchySchwartz}
\left( \sum \sqrt{m_k}\;q_k \cdot {p_k\over\sqrt{m_k}} \right)^2
\le 
	\left( \sum m_k |q_k|^2 \right)
	\left( \sum \frac{|p_k|^2}{m_k} \right)
	=2IT,
\end{equation}
 to the right hand side of equation (\ref{4ITMinus2B}) leads to
\[
4IT-2B \le 4IT,
\]
therefore $B \ge 0$.
Notice that the equality in the Cauchy-Schwartz inequality (\ref{CauchySchwartz}), and therefore the value $B=0$,
take place if and only if there exists a scalar $\lambda$, which may depend on time, such that $p_k/\sqrt{m_k} = \lambda \sqrt{m_k}\;q_k$, a relationship equivalent to
\begin{equation}
\label{eqForExpansiveMotion}
p_k = \lambda m_kq_k.
\end{equation}

We can now prove the following result.

\begin{lemma}
\label{lemmaMotionWithBeqZero}
For any solution of constant configurational measure of the planar, $\mathbb{R}^2$, or spatial, $\mathbb{R}^3$, $n$-body problem,
$B=0$ if and only if the motion is homothetic.
\end{lemma}
\begin{proof}
As shown above, if $B=0$, relation (\ref{eqForExpansiveMotion}) takes place. Integrating equation (\ref{eqForExpansiveMotion}), we obtain
\begin{equation}
\label{eqSimultaneousDilatation}
q_i(t)=\exp\left( \int_0^t \lambda dt \right) q_i(0),
\end{equation}
which means that the motion is homothetic.
Conversely, if the motion is homothetic, relation (\ref{eqForExpansiveMotion}) takes place, so $B=0$.
\end{proof}

For the Newtonian case, $a=1$,
we have $\Phi(I)=-4HI -4\mu \sqrt{I}$
with $\mu > 0$ (see Figure \ref{figPhi}).
\begin{figure}[htbp] %  figure placement: here, top, bottom, or page
   \centering
    \includegraphics[width=2.3in]{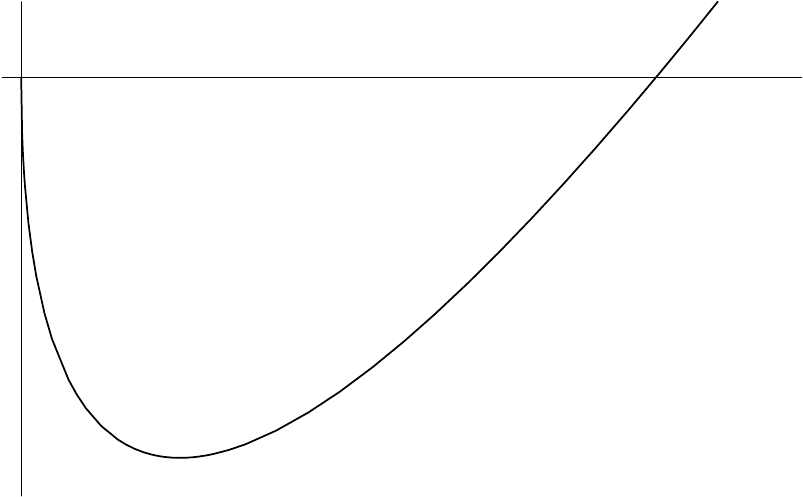}
    \includegraphics[width=2.3in]{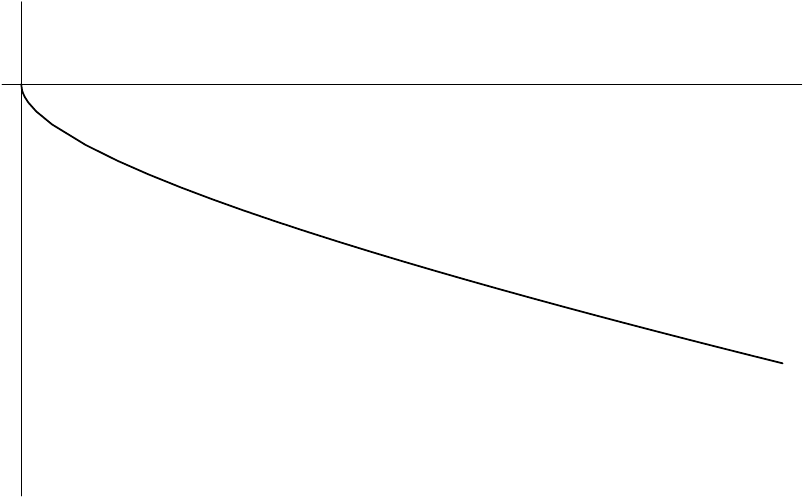} 
   \caption{The graph of $\Phi(I)$ in the Newtonian case
with $H<0$ (left) and $H \ge 0$ (right). The qualitative behaviour of $\Phi(I)$ is similar for all values of $a$ with $0 < a < 2$ but very different for $a=2$ or $a>2$.}
   \label{figPhi}
\end{figure}

Note that the mutual distances, $r_{jk}$, of orbits with $\mu=$ constant  and $a>0$ satisfy the inequality
\begin{equation}
\label{inequality}
\left(\frac{m_j m_k}{a\mu}\right)^{2/a} I
\le r_{jk}^2 \le
\frac{M}{m_j m_k}I,
%\label{ineq}
\end{equation}
which is derived from the relationships $m_j m_k r_{jk}^2\le M I $
and $a\mu I^{-a/2} = aU \ge m_j m_k/r_{jk}^a$.
%%%%

We can now prove the following lemma.
\begin{lemma}
\label{collisionIsTotal}
If $\mu=\mbox{constant}$ and $a>0$, then every collision is a total collision.
\end{lemma}
\begin{proof}
At any collision, at least one of the mutual distances, $r_{jk}$, must tend to zero. Then, according to inequality (\ref{inequality}), the moment of inertia must tend to zero ($I \to 0$), a fact that implies total collision.
\end{proof}
Similarly, if one of the mutual distances, $r_{jk}$, tends to infinity,
inequality (\ref{inequality}) implies that $I$ and all mutual distances tend to infinity.

We can now categorize the orbits with $\mu=$ constant and $0<a<2$ as follows:
\begin{itemize}
\item[(a)] $B=0$: The forward orbit or the backward orbit has a total collision, so $I \to 0$.
	We will treat this case in Section \ref{secCollisionOrbit}.
\item[(b)] $B>0$ and $H\ge 0$:
	There is a solution, $I_{\text{min}}$, of $\Phi(I)= -2B$ for which $dI/dt=0$. According to (\ref{inequality}), 
	$r_{jk}^2 \ge \left(\frac{m_j m_k}{a\mu}\right)^{2/a} I_{\text{min}}$.
	For $t \to \infty$, $I \to \infty$ and $r_{jk} \to \infty$.
	We will treat this case in Sections \ref{one} and \ref{two}.
\item[(c)] $B>0$ and $H< 0$:
	There are two solutions, $I_{\text{min}}$ and $I_{\text{max}}$, of $\Phi(I)= -2B$. The moment of inertia, $I$, oscillates between these two values. By (\ref{inequality}), the mutual distances are bounded from above and below,
	$\left(\frac{m_j m_k}{a\mu}\right)^{2/a}I_{\text{min}}
	\le r_{jk}^2 
	\le \frac{M}{m_j m_k}I_{\text{max}}$.
	This case will be treated in Section \ref{final}.
\end{itemize}

%%%%%%%%%%%%%%%%%%%%
%%%%%%%%%%%%%%%%%%%%
%%%%%%%%%%%%%%%%%%%%%
%%%%%%%%%%%%%%%%%%%%%
%%%%%%%%%%%%%%%%%%%%
\section{Collision orbits}
\label{secCollisionOrbit}
%%%%%%%%%%%%%%%%%%%%
%%%%%%%%%%%%%%%%%%%%%
%%%%%%%%%%%%%%%%%%%%%
%%%%%%%%%%%%%%%%%%%%
In this section, we will show that Saari's homographic conjecture is true for orbits that encounter collisions in the future or in the past.

\begin{theorem}
\label{propositionForTotalCollisionOrbit}
For the $n$-body problem given by a potential with $a<2$, every total collision solution of constant configurational measure is homothetic. 
\end{theorem}
\begin{proof}
For $a<2$, consider a solution that experiences a total collision forwards or backwards in time and for which $\mu=$ constant.
Since
\[
\frac{1}{2}
	\left(
		\frac{dI}{dt}
	\right)^2
	+
	\left(
		-4HI + 4\mu I^{(2-a)/2}
	\right)
	= -2B
\]
and
\[
\lim_{I \to 0} 
	\left(
		-4HI + 4\mu I^{(2-a)/2}
	\right)
	=0
\mbox{ for } a<2,
\]
the total collision, which occurs when $I \to 0$, can take place only if $-2B \ge 0$.
But we already know that $B \ge 0$.
Thus, for $a<2$, the total collision can take place if and only if
$B=0$. By Lemma \ref{lemmaMotionWithBeqZero}, the orbit is homothetic.
\end{proof}

An obvious consequence of Theorem  \ref{propositionForTotalCollisionOrbit} and Lemma \ref{collisionIsTotal} is the following result, which shows that, for appropriate homogeneous potentials, Saari's homographic conjecture is true for collision orbits in general.

\begin{theorem}
\label{collisionOrbitIsHomothetic}
For $0<a<2$, any collision orbit with $\mu$=constant
is homothetic.
\end{theorem}

We can now also show that the extended conjecture is true in the rectilinear (one-dimensional) case.

\begin{theorem}
\label{rectilinearCase}
For $0<a<2$, every rectilinear solution of $n$-body problem that has constant configurational measure is homothetic.
\end{theorem}
\begin{proof}
We can easily show (see e.g.\ \cite{collinear}) that every $n$-body rectilinear orbit with attractive force performs a collision in the future or in the past. Then Theorem \ref{collisionOrbitIsHomothetic} provides the proof.
\end{proof}

%%%%%%%%%%%%%%%%%%%
%%%%%%%%%%%%%%%%%%%
%%%%%%%%%%%%%%%%%%%
\section{Fujiwara coordinates}
%%%%%%%%%%%%%%%%%%%
%%%%%%%%%%%%%%%%%%%
%%%%%%%%%%%%%%%%%%%
\label{secDecomposition}
In this section, we introduce some new coordinates, which were proposed by one of us. They will help us describe the shape of the
planar configuration formed by the point masses.
%Equations (\ref{decompositionOfVelocity}) and (\ref{decompositionOfKineticEnergy}) are very important because they give the  explicit expression of the velocity's decomposition
%and provide explicit values of the vector's magnitudes.

%%%%%%%%%%%%
%%%%%%%%%%%%
We employ complex variables to express our dynamical variables in $\mathbb{R}^2$. For $(a_x, a_y)$, $(b_x, b_y) \in \mathbb{R}^2$,
we write $a=a_x+ia_y$, $b=b_x+ib_y  \in \mathbb{C}$.
We use the following notations: $a^\dag = a_x -ia_y$ for the complex conjugate, $a\cdot b=a_x b_x+a_y b_y$ for the inner product, and
$a\wedge b=a_x b_y - a_y b_x$ for the outer product.
Then we have
\[
a^\dag b = a\cdot b + i a\wedge b.
\]
%%%%%%%%%%%%

Valid in the planar $n$-body problem for any $n\ge3$, the Fujiwara coordinates, $Q_k$, are defined by
\begin{equation}
\label{defOfQ}
Q_k = \exp \left(-iC \int_0^t \frac{dt}{I}\right)
		\frac{q_k}{\sqrt{I}}.
\end{equation}
The scaling factor, $1/\sqrt{I(q)}$, makes the Fujiwara coordinates independent of size, and the phase factor $\exp (-iC\int dt/I)$ makes them also independent of the angular momentum.
Indeed, we can easily check that
\begin{equation}
\label{IforQ}
\sum m_k |Q_k|^2=1
\end{equation}
and 
\begin{equation}
\label{sumOfQP}
\sum m_k Q_k\cdot \frac{dQ_k}{dt}
+i\sum m_k Q_k\wedge \frac{dQ_k}{dt}
=\sum m_k Q_k^\dag \frac{dQ_k}{dt}=0.
\end{equation}
So in terms of Fujiwara coordinates, the moment of inertia and the angular momentum become $I(Q)=1$ and $C(Q)=0$.

Notice that relative to the coordinates $q_k$, the particle system is rotating as a whole with angular velocity $C/I(q)$. In Fujiwara coordinates, the phase factor is chosen such that an observer
placed in those coordinates rotates with the system at the same angular velocity  $C/I(q)$. So for the observer, the angular velocity appears to be zero. 

All the above facts make us expect that the evolution of the particle system expressed in $Q_k(t)$-coordinates describes the change in shape of the configuration given by the original variables $q_k(t)$. Indeed, we have the following result.
\begin{lemma}
\label{lammaHomographic1}
A planar solution expressed in coordinates $q_k(t)$ is homographic if and only if $dQ_k(t)/dt =0$ for all $k$ and $t$.
\end{lemma}
\begin{proof}
We can rewrite the definition of a homographic solution given in Section \ref{notations} as
\[
q_k(t)=R(t)e^{i\theta(t)}q_k(0),
\]
where $R \ge 0$ describes the evolution of the particle system in size and $\theta \in \mathbb{R}$ describes its rotation. Without loss of generality, we can take as initial conditions $R(0)=1$ and $\theta(0)=0$.
Then
\[
I(t) = \sum m_k |q_k|^2= R(t)^2 I(0)
\]
and
\[
C = R(t)^2 \frac{d\theta}{dt}I(0)=I(t)\frac{d\theta}{dt}.
\]
Thus, the Fujiwara variables take the form
\[
Q_k(t) 
= \exp{\left(-iC\int_0^t \frac{dt}{I}\right)}\frac{q_k(t)}{\sqrt{I}}
= e^{-i\theta(t)}\frac{ q_k(t)}{R(t)\sqrt{I(0)}}
=\frac{q_k(0)}{\sqrt{I(0)}},
\]
which is constant. Therefore, for any homographic orbit, 
$dQ_k/dt=0$.

Conversely, if $dQ_k/dt=0$, then $Q_k(t)=Q_k(0)$. So the variables $q_k(t)$ satisfy the relation
\[
q_k(t)
=\sqrt{I(t)}\exp{\left( iC\int \frac{dt}{I} \right)}Q_k(0).
\]
This implies that the solution is homographic.
\end{proof}

Using the fact that
\[
\frac{dQ_k}{dt}
=\frac{1}{\sqrt{I}}\exp{\left(-iC\int \frac{dt}{I}\right)}
	\Bigg(
		\frac{dq_k}{dt}
		+
		\left(
			-\frac{1}{2I}\frac{dI}{dt}-i\frac{C}{I}
		\right)
		q_k
	\Bigg),
\]
we get the decomposition of
the velocity and the kinetic energy in terms of Fujiwara coordinates,
\begin{equation}
\label{decompositionOfVelocity}
\begin{split}
\frac{dq_k}{dt}
&=	\left(
		i \frac{C}{I}
		+
		\frac{1}{2I}\frac{dI}{dt}
	\right)
	q_k
	+
	\sqrt{I}
	\exp{\left(iC\int \frac{dt}{I}\right)}
	\frac{dQ_k}{dt}\\
&=\sqrt{I}
	\exp{\left(iC\int \frac{dt}{I}\right)}
	\left\{
		\left(
			i \frac{C}{I}
			+
			\frac{1}{2I}\frac{dI}{dt}
		\right)
		Q_k
		+\frac{dQ_k}{dt}
	\right\},
\end{split}
\end{equation}

\begin{equation}
\label{decompositionOfKineticEnergy}
T(q)
=	\frac{C^2}{2I}
	+\frac{1}{8I}\left( \frac{dI}{dt}\right)^2
	+\frac{I}{2}\sum m_k \left| \frac{dQ_k}{dt} \right|^2.
\end{equation}
To obtain equation (\ref{decompositionOfKineticEnergy})
from equation (\ref{decompositionOfVelocity}), we have used the following ``orthogonality'' relationships between the basis vectors
$i Q_k$, $Q_k$, and $dQ_k/dt$, $k=1,\dots, n$,
\[
\sum m_k Q_k \wedge Q_k =0,
\sum m_k Q_k \cdot \frac{dQ_k}{dt}=0,
\sum m_k Q_k \wedge \frac{dQ_k}{dt}=0
\]
and the ``normalisation''
\[
\sum m_k |Q_k|^2=1.
\]
Fujiwara coordinates are strictly related to  Saari's decomposition of velocities \cite{saari2, saari-book} for the $n$-body problem. Using such decomposition one can express  the velocity as a sum of three orthogonal (at least in the coplanar case) components, namely 
$v = w_1 +  w_2 + w_3$, where $w_1$ describes changes in orientation for the system of particles,  $w_2$ changes in size, and $w_3$ changes in shape.  Equations (\ref{decompositionOfVelocity}) and (\ref{decompositionOfKineticEnergy}) are very important because they give the  explicit expression of the velocity's decomposition
and provide explicit values of the vector's magnitudes.

The geometrical meaning of the coordinates we use here can be expressed in terms of the so called {\itshape shape sphere}
 \cite{moeckel88, montgomery}.
If one reduces the configuration space by rotations and translations the three-body problem in the plane has a reduced configuration space isomorphic to $\mathbb{R}^3$. This reduced space is endowed with a metric, induced from the mass metric on the configuration space, which makes it a cone over a sphere of radius $\frac 1 2$.  This sphere is called the shape sphere and  is to be thought of as the space of oriented similarity classes of triangles \cite{montgomery}.
The coordinates we introduce below are coordinates on the fiber bundle obtained by the following operations: projecting  the motion on the shape sphere %(complex projective shape space if $n\ge > 4$) 
and then lifting uniquely (starting from $q_k(0)/I^{1/2}(0)$) the  result horizontally, for the connexion defined by the Saari's decomposition  of velocities, on the fiber  bundle whose total space is $\mathbb{R}^4\setminus\{0\}$ (i.e. the configuration  space modulo translations), base the shape sphere, fiber $\mathbb{C}\setminus\{0\}$ and  projection the quotient by the non trivial similitudes.

%In 1985, Saari \cite{saari2, saari-book} introduced an important decomposition of the velocity of the $n$-body problem in three orthogonal components. He expressed the velocity as 
%$v = w_1 +  w_2 + w_3$, where $w_1$ describes changes in orientation for the system of particles,  $w_2$ changes in size, and $w_3$ changes in shape. Equations (\ref{decompositionOfVelocity}) and (\ref{decompositionOfKineticEnergy}) are very important because they give the  explicit expression of the velocity's decomposition
%and provide explicit values of the vector's magnitudes.

The equations of motion expressed in Fujiwara variables are
\begin{equation}
\label{equationOfMotionForQ0}
\begin{split}
m_k \frac{d^2 Q_k}{dt^2}
=\frac{g_k(Q)}{I^{(a+2)/2}}
	&-\left(
		\frac{1}{I}\frac{dI}{dt}
		+i \frac{2C}{I}
	\right)
	m_k \frac{dQ_k}{dt}\\
	&-\left(
		\frac{1}{2I}\frac{d^2I}{dt^2}
		-\frac{1}{4I^2}\left(\frac{dI}{dt}\right)^2
		-\frac{C^2}{I^2}
	\right)
	m_k Q_k,
\end{split}
\end{equation}
where
$g_k(Q)=\sum_{j\ne k} m_j m_k (Q_j-Q_k)/r_{jk}^{a+2}(Q)$
and $I=I(q)$.

Before closing this section, let us give a proof of Saari's homographic conjecture in the collinear $n$-body case.
\begin{theorem}
Any collinear $n$-body orbit with non-vanishing angular momentum
is homographic.
\end{theorem}
The proof of this theorem has been given by three of this paper's authors in \cite{collinear} for any potential that depends only on the mutual distances. Alternatively, Saari also proved this theorem in \cite{saari-book}. In the context of Fujiwara coordinates, the proof goes as follows.
\begin{proof}
For collinear motion, we can write that
\[
q_k(t)=e^{i\phi(t)}r_k(t),
\]
with $\phi, r_k \in \mathbb{R}$ and $\phi(0)=0$.
Then,
\[
I=\sum m_k |q_k|^2 = \sum m_k r_k^2,
\]
\[
C
=\Im \left( \sum m_k q_k^\dag \frac{dq_k}{dt} \right)
=I \frac{d\phi}{dt},
\]
where $\Im (z)$ is the imaginary part of the complex number $z$.
Therefore
\[
Q_k(t)
=\exp\left( -i C \int_0^t \frac{dt}{I} \right)
	e^{i\phi}
	\frac{r_k}{\sqrt{I}}
=\frac{r_k}{\sqrt{I}}
\in \mathbb{R}.
\]
So, $dQ_k/dt, d^2Q_k/dt^2$, and $g_k(Q)$ are real.
Then, in equation (\ref{equationOfMotionForQ0}) with $C\ne 0$,
only one term,
\[
-i \frac{2C}{I}
	m_k \frac{dQ_k}{dt},
\]
is purely imaginary and all the other terms are real.
Thus $dQ_k/dt=0$.
By Lemma \ref{lammaHomographic1},
the motion is homographic.

To prove this theorem for any potential that depends only on the mutual distances, as it was shown in \cite{collinear}, it is enough to
replace the term $g_k/I^{(a+2)/2}$ in equation (\ref{equationOfMotionForQ0}) with a suitable function that is real for real values of $Q_k$.
\end{proof}

%%%%%%%%%%%%%%%%%%%
%%%%%%%%%%%%%%%%%%%
%%%%%%%%%%%%%%%%%%%

\section{Shape evolution for constant configurational measure solutions}
\label{kineticSection}
%%%%%%%%%%%%%%%%%%%
%%%%%%%%%%%%%%%%%%%
%%%%%%%%%%%%%%%%%%%
In this section, we prove two lemmas that characterise the solutions with $\mu=$ constant. Then we provide a simple expression of the kinetic energy and obtain the equations of motion that describe shape evolution. 

Substituting equations (\ref{eqIntegralForI}) and (\ref{valueOfPhi})
into equation (\ref{decompositionOfKineticEnergy}), we obtain
\begin{equation}
\label{kineticEnergyForShapeOriginal}
\sum m_k \left| \frac{dQ_k}{dt} \right|^2
=\frac{B-C^2}{I^2(q)}.
\end{equation}
Since the left hand side must be non-negative, $B$ is non-negative and
\begin{equation}
B\ge C^2 \ge 0.
\end{equation}
Note that the fact that $B-C^2\geq 0$ is nothing but  the well known Sundman inequality (see \cite{AC,saari-book}) and the cases of equality are known to be homographic motions whose  configuration is central \cite{AC, saari-book}.
We thus have the following result.
\begin{lemma}
\label{lammaHomographic2}
Every planar solution of constant configurational measure is homographic if and only if $B=C^2$.
\end{lemma}
\begin{proof}
The statement is obvious by equation (\ref{kineticEnergyForShapeOriginal})
and Lemma \ref{lammaHomographic1}.
\end{proof}

Equation (\ref{kineticEnergyForShapeOriginal}) shows that
if the configurational measure is constant,
then $I^2(q) \sum m_k |dQ_k/dt|^2 = B-C^2=$ constant.
Saari showed that the converse is also true. He proved that
$I^2(q) \sum m_k |dQ_k/dt|^2 =$ constant if and only if $\mu=$constant \cite{saari-book}.
Indeed, he derived a nice relation between
$I^2(q) \sum m_k |dQ_k/dt|^2$ and $\mu=U(q) I^{a/2}(q)=U(Q)$.
In our notation, this relation is
\[
\frac{d}{dt} 
	\left(
		I^2(q) \sum m_k \left|\frac{dQ_k}{dt}\right|^2
	\right)
	=
	2I^{1-a/2}(q) \frac{d\mu}{dt}.
\]
To prove this relation, let us take the innner product with $dQ_k/dt$ and equation (\ref{equationOfMotionForQ0}), and we obtain
\[
\sum m_k \frac{dQ_k}{dt}\cdot\frac{d^2Q_k}{dt^2}
=\frac{1}{I^{(a+2)/2}} \sum g_k(Q)\cdot\frac{dQ_k}{dt}
	-\frac{1}{I}
	\sum m_k \left| \frac{dQ_k}{dt} \right|^2.
\]
Therefore
\[
\begin{split}
\frac{d}{dt} 
	\left(
		I^2 \sum m_k \left|\frac{dQ_k}{dt}\right|^2
	\right)
&=2 I^{1-a/2} \sum g_k(Q)\cdot \frac{dQ_k}{dt}\\
&=2I^{1-a/2} \frac{dU(Q)}{dt}.
\end{split}
\]

The equations of motion in Fujiwara coordinates 
for constant configurational measure solutions are given by
substituting equations (\ref{eqPhi}), (\ref{eqMotForI}), and (\ref{eqIntegralForI})
into (\ref{equationOfMotionForQ0}).
We thus obtain
%%%%%%%%%%
\begin{equation}
\label{equationOfMotionForQ}
m_k \frac{d^2 Q_k}{dt^2}
=	\left(
	-\frac{1}{I}\frac{dI}{dt}
	-2i\frac{C}{I}
	\right)m_k \frac{dQ_k}{dt}
	+
	\frac{G_k(Q)}{I^{(a+2)/2}}
	-\frac{B-C^2}{I^2}m_k Q_k,
\end{equation}
where
\begin{equation}
\label{defOfG}
G_k(Q) = g_k(Q)+a \mu m_k Q_k,
\end{equation}
$I=I(q)$, and $g_k(q)=\sum_{j\ne k} \frac{m_jm_k}{r_{jk}^{a+2}}(q_j-q_k)=m_k \frac{d^2 q_k}{dt^2}$ (see the equations of motion in Section \ref{notations}).

The configurations of the particle system for which $G_k(Q)=0$ are called {\it central configurations}. It is well known that for the planar three-body problem there are five classes of central configurations
(up to rotations and dilations/contractions): two equilateral (one for each orientation the triangle that has the point masses at its vertexes) and three rectilinear (one for each ordering of the bodies on a line) in which the ratio of the distances between particles is given by a complicated formula (obtained by Euler) that involves the values of the masses. 

We can now prove the following lemma
\begin{lemma}
\label{lammaHomographic3}
In the planar $n$-body problem, if the orbit is homographic, then the the bodies maintain the shape of the same central configuration for all times.
\end{lemma}
%%%
\begin{proof}
From Lemmas \ref{lammaHomographic1} and \ref{lammaHomographic2}, $dQ/dt=0$, $d^2 Q/dt^2=0$, and $B=C^2$ for every homographic motion. Thus by equation (\ref{equationOfMotionForQ}), $G_k=0$. This implies that the bodies maintain the shape of the same central configuration for all time.
\end{proof}

The number of central configurations is known to be finite
for arbitrary masses  only in the three- and four-body problem.
In the planar three-body problem, the number of central configurations is five. This fact was proved by Moulton \cite{moulton} for $a>-1$ and by Albouy \cite{albouy} for $a>-2$. Hampton and Moeckel \cite{HamptonMoeckel} recently showed that the number of central configurations in the Newtonian planar four body problem is finite.
For these cases, if the motion satisfies $G_k=0$ for all time,
the bodies maintain the shape of the same central configuration,
since there is only a finite number of central configurations and
because once a solution forms a central configuration it cannot switch to another one due to continuity reasons. Therefore, the motion is homographic. Thus, we obtained the following result.

\begin{lemma}
\label{lammaHomographic4}
In the planar three-body problem for potentials with $a>-2$ and in the 
planar four-body problem given by the Newtonian potential, if $G_k=0$ for all time, then the solution is homographic and  the bodies maintain the shape of the same central configuration.
\end{lemma}

Let us further define the quantity
\[
\rho
=\sqrt{
	\frac{m_1 m_2 m_3}{M}
	\sum \frac{|G_\ell(Q)|^2}{m_\ell}
	},
\]
which in a certain sense measures the magnitude of the mathematical object formed by all functions $G_k$.

Relative to the possible values of $B-C^2$ and $\rho$, there are four cases to discuss:
\begin{enumerate}
\item[(i)]	$B-C^2=0$, $\rho=0$, 
\item[(ii)]	$B-C^2=0$, $\rho \ne 0$, 
\item[(iii)]	$B-C^2>0$, $\rho=0$, 
\item[(iv)]	$B-C^2>0$, $\rho \ne 0$.
\end{enumerate}
Lemmas \ref{lammaHomographic1}, \ref{lammaHomographic2}, and \ref{lammaHomographic3} state that if $B-C^2=0$, then $\rho=0$.
Therefore (i) is possible, whereas (ii) is impossible.
According to Lemmas \ref{lammaHomographic1}, \ref{lammaHomographic2}, and \ref{lammaHomographic4},
if $\rho=0$ then $B-C^2=0$ for the problems stated in Lemma \ref{lammaHomographic4}. Therefore, (i) is possible, whereas (iii) is impossible in these cases.

Saari's homographic conjecture states that the only possible solutions with $\mu=$ constant occur in case (i), and that cases (ii), (iii), and (iv) are not realised.
Therefore, Lemmas \ref{lammaHomographic1}, \ref{lammaHomographic2}, 
\ref{lammaHomographic3}, and \ref{lammaHomographic4} 
allow us to state Saari's homographic conjecture 
of the planar three-body problem as follows.
%%%%%%%%%%%%%%%%%%%%%%
\begin{conjecture}[Saari's homographic]
\label{conjecture2}
There are no  solutions of the planar three-body problem
with $\mu=$ constant, $B-C^2>0$, and $\rho \ne 0$.
\end{conjecture}

As shown above, Conjecture \ref{conjecture1} and Conjecture \ref{conjecture2} are equivalent in the planar 3-body problem for homogeneous potentials with $a>-2$ and in the planar 4-body 
problem for the Newtonian potential.
In the following sections, we will investigate the properties of the planar three-body motion with $\mu=$ constant, $B-C^2>0$, and $\rho \ne 0$.

Before closing this section, let us simplify our equations. For this we consider a fictitious (scaled) time variable, $\tau$, defined as
\begin{equation}
\label{defOfTau}
d\tau = \frac{dt}{I(q)},
\end{equation}
and introduce what we call Fujiwara momenta by taking
\begin{equation}
P_k = m_k \frac{dQ_k}{d\tau} = I(q) m_k \frac{dQ_k}{dt}.
\end{equation}
These simplifications are suggested by equations
(\ref{defOfQ}) and (\ref{kineticEnergyForShapeOriginal}).
Indeed, for solutions with $\mu=$ constant, the corresponding Fujiwara kinetic energy becomes
\begin{equation}
\label{kineticEnergyForShape}
\sum \frac{|P_k|^2}{m_k}
=B-C^2
= {\rm constant},
\end{equation}
and the equations of motion take the form
\begin{equation}
\label{equationOfMotion}
\frac{dP_k}{d\tau}
=-2iC P_k
	+I(q)^{(2-a)/2}G_k(Q)
	-(B-C^2)m_k Q_k.
\end{equation}

%%%%%%%%%%%%%%%%%%%
%%%%%%%%%%%%%%%%%%%
%%%%%%%%%%%%%%%%%%%
\section{Candidates for non-homographic solutions}
\label{candidate}
%%%%%%%%%%%%%%%%%%%
%%%%%%%%%%%%%%%%%%%
%%%%%%%%%%%%%%%%%%%
Since our goal is to show that constant configurational measure solutions of the planar three-body problem are homographic, we will seek candidates for constant configurational measure non-homographic solutions, aiming to prove that they don't exist. 
Some simple algebra will show that for each such candidate in Fujiwara coordinates, $Q_k$, there are only two possible Fujiawara momenta, $P_k$ and $-P_k$. We start with the following result.

\begin{lemma}
\label{lemmaSST}
If six given complex quantities, $\xi_k$ and $\eta_k, k=1,2,3,$ satisfy the properties $\sum \xi_k^\dag \eta_k=0$, $\sum \eta_k =0$, and $\sum m_j m_k |\xi_j - \xi_k|^2>0$, then there is a complex number $\zeta$,
such that 
$\eta_\ell = \zeta (\xi_j^\dag - \xi_k^\dag)$.
\end{lemma}
\begin{proof}
Since $\sum m_j m_k |\xi_j - \xi_k|^2>0$,
at least two of the quantities $\xi_j - \xi_k$ are not zero, say $\xi_2-\xi_3\ne 0$ and $\xi_3 - \xi_1\ne0$.
From $\eta_3=-\eta_1-\eta_2$, it follows that
$0=\sum \xi_k^\dag \eta_k=(\xi_1^\dag-\xi_3^\dag)\eta_1+(\xi_2^\dag-\xi_3^\dag)\eta_2$.
Therefore,
\[
\frac{\eta_1}{\xi_2^\dag-\xi_3^\dag}=\frac{\eta_2}{\xi_3^\dag-\xi_1^\dag}.
\]
Then the value of this ratio is the number $\zeta$ we are seeking. Indeed, we have
$\eta_1=\zeta(\xi_2^\dag-\xi_3^\dag)$ and $\eta_2=\zeta(\xi_3^\dag-\xi_1^\dag)$,
so
$\eta_3=-(\eta_1+\eta_2)$ yields
$\eta_3=\zeta(\xi_1^\dag-\xi_2^\dag)$.
\end{proof}

A geometrical interpretation of this lemma is as follows \cite{ffkoySST}.
For the three-body problem, in coordinates having the origin at the centre of mass of the particle system (which implies that $\sum p_k=0$), if the moment of inertia is constant and the angular momentum is zero (i.e.\ $\sum q_k^\dag p_k=0$) and no triple-collision occurs (i.e.\ $I>0$), then the triangle whose vertices are $q_1, q_2, q_3$ and the triangle whose perimeters are $p_1, p_2, p_3$
are inversely similar (i.e.\ there exists a $\zeta$ such that $p_\ell=\zeta (q_j^\dag-q_k^\dag$)).
%%%%%%%%

Since $\sum Q_k^\dag P_k =0$, $\sum P_k=0$, and $I(Q)=1$,
Lemma \ref{lemmaSST} applies to the variables $Q_k$ and $P_k$.
Therefore there exist a non-negative variable $\kappa$ and a
real variable $\phi$ such that
\begin{equation}
\label{SSTforQP}
P_\ell = \kappa e^{i\phi} (Q_j^\dag-Q_k^\dag),
\end{equation}
with $(j, k, \ell) = (1,2,3)$, $(2,3,1)$, $(3,1,2)$.
From the fact that
\[
\sum \frac{|P_\ell|^2}{m_\ell}
	= \frac{\kappa^2}{m_1 m_2 m_3}
		\sum m_j m_k |Q_j-Q_k|^2
	= \frac{M\kappa^2}{m_1 m_2 m_3},
\]
we can obtain the value for $\kappa$, which turns out to be constant. Indeed,
\begin{equation}
\label{defOfKappa}
	\kappa
		=\sqrt{
			\frac{m_1 m_2 m_3}{M}
			\sum \frac{|P_\ell|^2}{m_\ell}
			}
		=\sqrt{
			\frac{m_1 m_2 m_3(B-C^2)}{M}
			}.
\end{equation}

It is easy to check that
$Q_k$ and $G_k$ also satisfy
\[
\sum Q_k^\dag G_k(Q)=0.
\]
Thus, by 
Lemma \ref{lemmaSST}, there exist a positive value $\rho$ and a real variable $\psi$ such that
\begin{equation}
\label{SSTforQG}
G_\ell(Q) = \rho e^{i\psi}(Q_j^\dag - Q_k^\dag),
\end{equation}
with
\begin{equation}
\label{defOfRho}
\rho
=\sqrt{
	\frac{m_1 m_2 m_3}{M}
	\sum \frac{|G_\ell|^2}{m_\ell}
	}.
\end{equation}

Combining (\ref{SSTforQP}) and (\ref{SSTforQG}), and assuming $\rho\ne0$, we obtain the following relationship between $P_k$ and $G_k$:
\begin{equation}
\label{SSTforPG}
P_k
=\frac{\kappa}{\rho}e^{i(\phi-\psi)}G_k.
\end{equation}
It is important to note that the scale factor $\kappa/\rho$ and the phase factor $\exp{i(\phi-\psi)}$ are the same for all $k=1,2,3$.

Now consider the condition
\[
\mu
=U(q)I(q)^{a/2}
=U(Q)
={\rm constant}.
\]
Differentiation with respect to $\tau$ yields
\[
\begin{split}
0
&=-\frac{dU(Q)}{d\tau}
=\sum \frac{1}{m_k}
	P_k \cdot g_k(Q)
=\sum \frac{1}{m_k}
	P_k \cdot G_k(Q)\\
&=\frac{\kappa}{\rho}
	\sum \frac{|G_k(Q)|^2}{m_k}
	\cos (\phi-\psi)\\
&=\sqrt{
	(B-C^2)
	\sum \frac{|G_k(Q)|^2}{m_k}
	}
	\; \cos (\phi-\psi).
\end{split}
\]
If the orbit is not homographic,
then $B \ne C^2$ and $G_k \ne 0$,
therefore,
$\cos (\phi-\psi)=0$.
From (\ref{SSTforPG}) we can then conclude that
\begin{equation}
\label{theCandidate}
P_k = i\epsilon \frac{\kappa}{\rho} G_k,
\mbox{ with }
\epsilon = \pm 1,
\kappa \ne 0,
G_k \ne 0. 
\end{equation}
We have thus proved the following result.

\begin{lemma}
\label{keylemma}
Consider a solution $q_k, k=1,2,3$, of the planar three-body problem that is non-homographic and has constant configurational measure, and let $Q_k, k=1,2,3$, be its corresponding Fujiwara coordinates.
Then the corresponding Fujiwara momenta must be of the form 
(\ref{theCandidate}), with $k=1,2,3$.
\end{lemma}

By Lemma \ref{keylemma}, our goal can be redefined as aiming to show that no Fujiwara momenta of the form (\ref{theCandidate}) can satisfy the equations of motion (\ref{equationOfMotion}).
%%%%%%%%%%%%%%%%%%%%%%%%%%%%%%%%%% 
%%%%%%%%%%%%%%%%%%%%%%%%%%%%%%%%%%
%%%%%%%%%%%%%%%%%%%%%%%%%%%%%%%%%%
\section{A condition for the non-homographic candidate}
\label{condition}
%%%%%%%%%%%%%%%%%%%%%%%%%%%%%%%%%%
%%%%%%%%%%%%%%%%%%%%%%%%%%%%%%%%%%
%%%%%%%%%%%%%%%%%%%%%%%%%%%%%%%%%%
In this section, we derive a condition for the non-homographic candidate (\ref{theCandidate}) to satisfy the equation of motion (\ref{equationOfMotion}). This condition, see (\ref{theCondition2}),
 will be later useful for proving our theorems about Saari's homographic conjecture.

By differentiating equation (\ref{SSTforQP}) with respect to $\tau$,
we have
\begin{equation}
\label{equationOfMotion2}
\begin{split}
\frac{dP_\ell}{d\tau}
&=\kappa e^{i\phi}
	\left\{
		i\frac{d\phi}{d\tau} (Q_j^\dag-Q_k^\dag)
		+
		\left(
			\frac{P_j^\dag}{m_j}-\frac{P_k^\dag}{m_k}
		\right)
	\right\}\\
&=i\frac{d\phi}{d\tau}P_\ell
	-(B-C^2)m_\ell Q_\ell.
\end{split}
\end{equation}
Comparing this equation with (\ref{equationOfMotion}),
we obtain the condition
\[
G_k 
	\left(
\epsilon \frac{d\phi}{d\tau}
		+2\epsilon C
		+\frac{\rho}{\kappa}I^{(2-a)/2}
	\right)
	\frac{\kappa}{\rho}
=0.
\]
Since we are analysing the motion with $G_k \ne 0$ and $\kappa \ne 0$, the necessary and sufficient condition for the candidate (\ref{theCandidate}) to satisfy the equation of motion (\ref{equationOfMotion}) is
\begin{equation}
\label{theCondition}
\epsilon \frac{d\phi}{d\tau}
		+2\epsilon C
		+\frac{\rho}{\kappa}I^{(2-a)/2}
=0.
\end{equation}

To have an explicit expression for $e^{i\phi}$, we consider the quantities
\[
\gamma_{jk}=m_jQ_jP_k - m_k Q_k P_j.
\]
Since $\sum P_k=0$ and $\sum m_k Q_k=0$,
we have $\gamma_{12}=\gamma_{23}=\gamma_{31}$.
With the help of equation (\ref{SSTforQP}), we obtain
\[
\gamma_{12}
=m_1Q_1P_2-m_2Q_2P_1
=-\kappa e^{i\phi}.
\]
Using equation (\ref{theCandidate}), it follows that
\[
\gamma_{12}
=i\epsilon \frac{\kappa}{\rho}
	\left(
		m_1Q_1G_2-m_2Q_2G_1
	\right)
=-i\epsilon \frac{\kappa}{\rho}
	m_1 m_2 m_3
	\sum \frac{1}{r_{jk}^{a+2}}(Q_j-Q_k)Q_\ell.
\]
Thus,
\[
e^{i\phi}
=i\epsilon \frac{m_1m_2m_3}{\rho}
	\sum \frac{1}{r_{jk}^{a+2}}(Q_j-Q_k)Q_\ell
\]
and
\begin{equation}
\begin{split}
\epsilon\frac{d\phi}{d\tau}
&= m_1m_2m_3 e^{-i\phi}
	\frac{d}{d\tau}
	\left(
		\frac{1}{\rho}
		\sum \frac{1}{r_{jk}^{a+2}}(Q_j-Q_k)Q_\ell
	\right).
\end{split}
\end{equation}
A straightforward computation, which uses the equality $m_k dQ_k/d\tau = i\epsilon\kappa G_k/\rho$ and equation (\ref{drij}) for $dr_{jk}/d\tau$, yields
\[
\begin{split}
\epsilon \frac{d\phi}{d\tau}
=(a+2)
	&\frac{m_1^2m_2^2m_3^2\Delta^2}{M}
	\frac{\kappa}{\rho^3}
	\sum \frac{m_\ell}{r_{jk}^{a+4}}
	\left(
		\frac{1}{r_{k\ell}^{a+2}}-\frac{1}{r_{\ell j}^{a+2}}
	\right)^2\\
	&+
	\frac{\kappa}{\rho}
	\left(
		2a\mu - \sum \frac{m_j+m_k}{r_{jk}^{a+2}}
	\right),
\end{split}
\]
where $\Delta$ is twice the oriented area of the triangle $Q_1Q_2Q_3$,
\[
\Delta=Q_1\wedge Q_2+Q_2\wedge Q_3+Q_3\wedge Q_1.
\]
Therefore condition (\ref{theCondition}) becomes
\begin{equation}
\label{theCondition2}
-I^{(2-a)/2}(q)
=\frac{m_1^2m_2^2m_3^2}{M^2}(B-C^2)
	\left(
		\frac{f_1}{\rho^4}+\frac{f_2}{\rho^2}
	\right)
	+
	2\epsilon C
	\sqrt{\frac{m_1m_2m_3}{M}(B-C^2)}
	\;\frac{1}{\rho},
\end{equation}
with
\[
f_1
=(a+2)m_1 m_2 m_3\Delta^2
	\sum \frac{m_\ell}{r_{jk}^{a+4}}
	\left(
		\frac{1}{r_{k\ell}^{a+2}}
		-
		\frac{1}{r_{\ell j}^{a+2}}
	\right)^2
\]
and
\[
f_2
=\frac{M}{m_1m_2m_3}
 	\left(
		2a\mu 
		-\sum \frac{m_j+m_k}{r_{jk}^{a+2}}
	\right).
\]
Note that the variables $r_{jk}$ and $\rho$ on the right hand side of equation (\ref{theCondition2}) are functions of $Q$ while the moment of inertia, $I$, on the left hand side depends on $q$.

It is important to remark that condition (\ref{theCondition2}) will be crucial in our further understanding of Saari's homographic conjecture. We will refer to it throughout the rest of this paper.

%%%%%%%%%%%%%%%%%%%
%%%%%%%%%%%%%%%%%%%
%%%%%%%%%%%%%%%%%%%
\section{Saari's homographic conjecture for non-negative energy, I}
\label{one}
%%%%%%%%%%%%%%%%%%%
%%%%%%%%%%%%%%%%%%%
%%%%%%%%%%%%%%%%%%%
In this section we will prove that Saari's homographic conjecture of the three-body problem is true for $0<a<2$ in the non-negative energy case if the solutions do not scatter to infinity in a particular way, namely by tending towards one of the three possible rectilinear central configurations. We do not claim that the conjecture is false and such motions occur, but at this point we cannot overcome the technical difficulties required to prove the result in general. (In Section \ref{two},
we will prove that such motions do not occur in the Newtonian case when all masses are equal.)

Recall from Section \ref{kineticSection} that there are three rectilinear central configurations in the three-body problem.
We will call them the rectilinear central configurations 1, 2, and 3,
in agreement with the index of the middle mass, namely
$(m_3,m_1,m_2)$, $(m_1,m_2,m_3)$, and $(m_2,m_3,m_1)$,
respectively. Each case corresponds to a certain value of the configurational measure; we call them the critical values of $\mu$ and denote them by $\mu_c^{(1)}$, $\mu_c^{(2)}$, and $\mu_c^{(3)}$.

For each $k=1,2,3$,  $\mu(Q)=\mu_c^{(k)}$ defines solutions that
pass through the rectilinear central configuration $k$.
We call such solutions ``critical paths,'' see Figure \ref{figCriticalPathsForGeneralMasses}, where we have taken
$q_1=(-1,0)$, $q_2=(1,0)$ and $q_3=(x,y)$. This can be done without loss of generality because $\mu(Q)$ is invariant under translation, rotation, and scaling of the original variables $q_k$.

\begin{figure}[htbp] 
   \centering
   \includegraphics[width=2.2in]{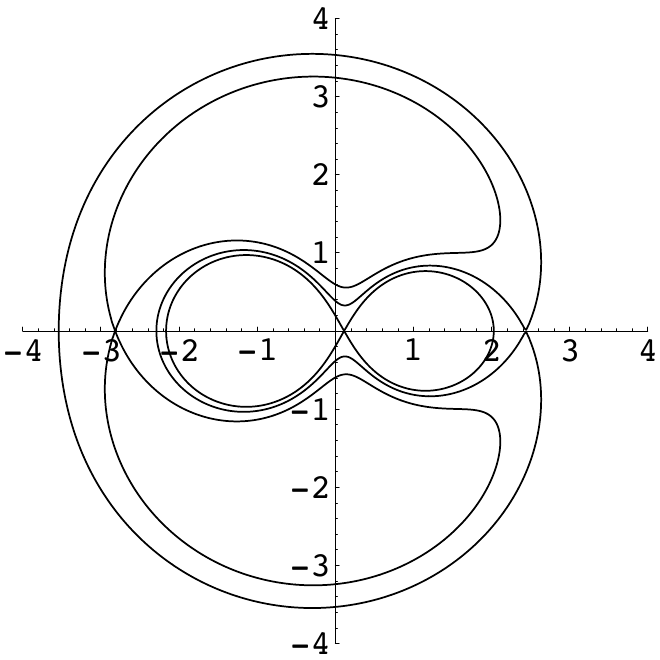}
   \hspace{0.2cm}
    \includegraphics[width=2.2in]{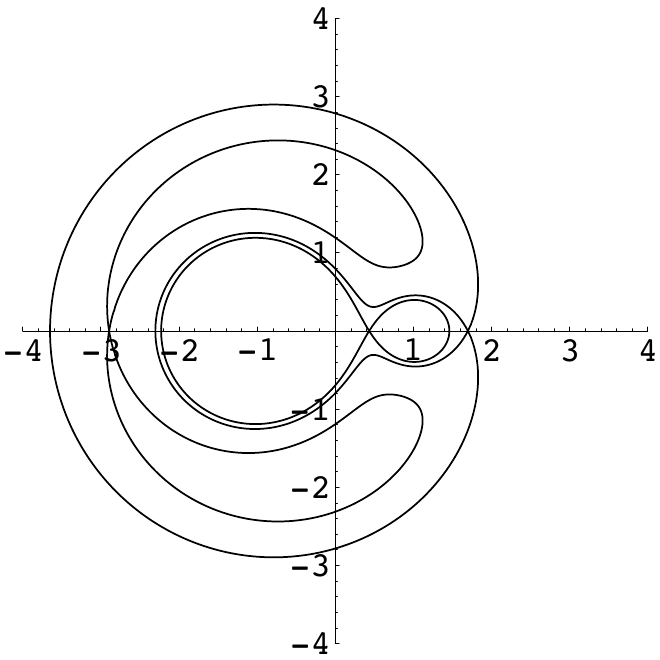} 
   \caption{The critical paths for $a=1$,
   $q_1=(-1,0)$, $q_2=(1,0)$ and $q_3=(x,y)$.
   The left: $m_1=4$, $m_2=2$, $m_3=1$.
   The right: $m_1=1000$, $m_2=100$, $m_3=1$.
   }
   \label{figCriticalPathsForGeneralMasses}
\end{figure}

We can now prove the main result of this section.

\begin{theorem}
\label{mainPropositionForNonNegativeEnergy}
In the planar three-body problem with non-negative energy for a potential satisfying $0<a<2$, Saari's homographic conjecture is true if the configuration is not on any of the critical paths given by  $\mu = \mu_c^{(k)}$.
\end{theorem}
\begin{proof}
From equations (\ref{eqPhi}), (\ref{eqMotForI}), and (\ref{eqIntegralForI})
for $0<a<2$ and $H\ge 0$, it follows that $I \to \infty$ when $t \to \infty$.
Then, from the left hand side of equation (\ref{theCondition2}), we conclude that $I^{(2-a)/2} \to \infty$.

We can omit now from our considerations all collision solutions.
Indeed, we already proved that Saari's homographic conjecture is true
for collision orbits for homogeneous potentials with $0<a<2$, see Theorem \ref{collisionOrbitIsHomothetic}.
Therefore, we can restrict our analysis to collision-free solutions.

By (\ref{inequality}), the mutual distances in Fujiwara coordinates are bounded from below and from above by constants,
\begin{equation}
\label{inequalityForQ}
\left(\frac{m_j m_k}{a\mu}\right)^{2/a}
\le r_{jk}^2(Q)=|Q_j-Q_k|^2 \le
\frac{M}{m_j m_k}.
\end{equation}
Then $f_1$ and $f_2$ are finite in the right hand side of equation (\ref{theCondition2}).
To make the right hand side of this equation infinite, let
$\rho \to 0$ (namely, $G_k \to 0$) for $t \to \infty$.
%%%%%%%%
Consequently every solution must asymptotically approach a central configuration.
Since the equilateral central configurations are isolated minima of the configurational measure (see e.g. \cite{wintner}), the orbit cannot tend to an equilateral triangle. Therefore, the only possibility is that the particle system tends to one of the rectilinear central configurations along a critical path. Thus the configuration must belong to one of the critical paths defined by a configurational measure that takes the value $\mu_c^{(k)}, k=1,2,3$. This completes the proof.
\end{proof}

An alternative proof of the theorem above was proposed by Alain Albouy and Alain Chenciner after they read an earlier version of this paper. The main advantage of their proof is  that it holds for any number of  bodies. Their proof is a direct consequence of Chazy's expansion \cite{Chazy}. We hope that Chenciner and Albouy will publish their proof somewhere else.

%%%%%%%%%%%%%%%%%%%
%%%%%%%%%%%%%%%%%%%
%%%%%%%%%%%%%%%%%%%
\section{Analytic behaviour near the rectilinear central configurations}
\label{nearTheCollinearAlongTheCriticalPath}
\label{break}
%%%%%%%%%%%%%%%%%%%
%%%%%%%%%%%%%%%%%%%
%%%%%%%%%%%%%%%%%%%
In this section, we consider the behaviour of the functions $f_1$, $f_2$, and $\rho$ near the rectilinear central configurations along the critical path in the Newtonian ($a=1$) three-body equal-mass case. Without loss of generality, we can take $m_1=m_2=m_3=1$.
Then the three critical paths are identical and have the same
configurational measure, $\mu(Q)=5/\sqrt{2}$ (see Figure \ref{figCriticalPath}).
\begin{figure}[htbp] %  figure placement: here, top, bottom, or page
   \centering
   \includegraphics[width=2in]{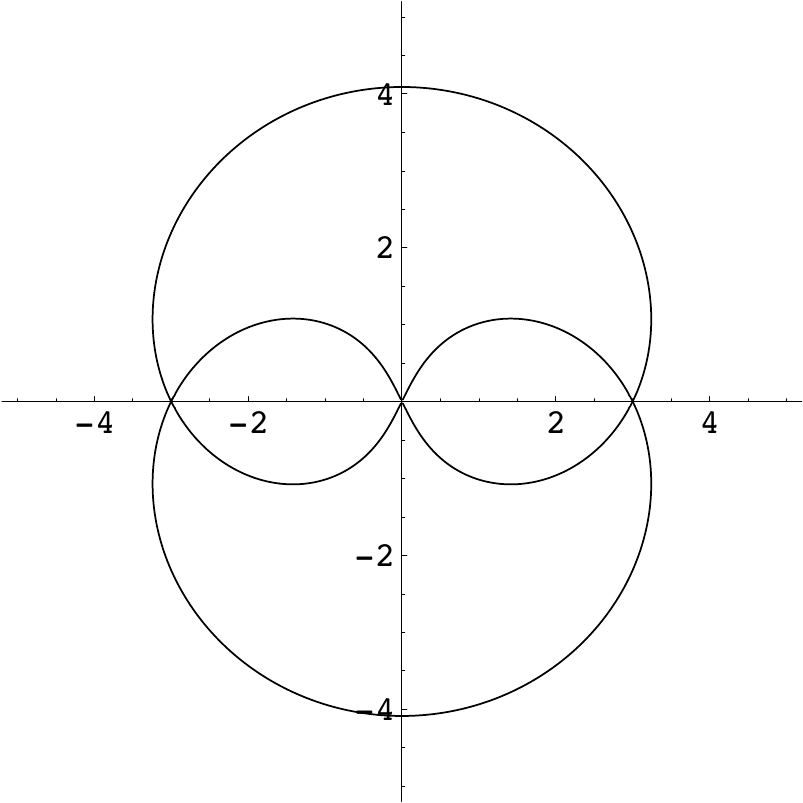} 
   \caption{The critical path for $m_1=m_2=m_3=1$, $a=1$,
   $q_1=(-1,0)$, $q_2=(1,0)$, and $q_3=(x,y)$.}
   \label{figCriticalPath}
\end{figure}

Let us observe that $f_1(Q)$, $f_2(Q)$, and $\rho(Q)$ are
functions of $r_{jk}(Q)=|Q_j-Q_k|=|q_j-q_k|/\sqrt{I(q)}$,
therefore depend on the shape of the triangle $q_1q_2q_3$,
where the shape is a similarity class of the triangle.
We say that two triangles belong to the same similarity class
if they have the same shape.
Therefore without loss of generality we can fix the shape of the triangle
by taking $q_1=(-1,0)$, $q_2=(1,0)$ and $q_3=(x,y)$.

Then 
\[
Q_k=q_k \sqrt{\frac{3}{6+2(x^2+y^2)}}
\]
and
\[
\mu(Q)
=\sqrt{\frac{3+x^2+y^2}{6}}
	\left(
		1+
		\frac{2}{\sqrt{(x-1)^2+y^2}}+
		\frac{2}{\sqrt{(x+1)^2+y^2}}
	\right).
\]
The three rectilinear central configurations are given by
$(x,y)=(-3,0)$, $(0,0)$, and $(3,0)$.
These configurations are mutually equivalent, so the behaviour near $(0,0)$ is the same as near $(\pm 3,0)$. Therefore we will further investigate only the behaviour near the origin.

Notice that in the neighbourhood of $(0,0)$, $\mu(Q)$ is of the form
\[
\mu(Q)
=\frac{5}{\sqrt{2}}
+\frac{29x^2-7y^2}{6\sqrt{2}}
+\frac{331x^4-850x^2 y^2 + 79y^4}{72\sqrt{2}}
+\mbox{ higher order terms}.
\]
Therefore the critical path $\mu(Q)=5/\sqrt{2}$ near the origin behaves like
\begin{equation}
\label{seriesFortheCriticalPath}
y^2
=\frac{29}{7}x^2-\frac{7491}{343}x^4+O(x^6).
\end{equation}
Observe that the function $\rho^2$ along the critical path is of the form
\begin{equation}
\label{rho2SeriesFortheCriticalPath}
\begin{split}
\rho^2(Q)
&=\frac{841x^2+49y^2}{18}
+\frac{11281x^4 - 7570x^2 y^2 - 455y^4}{54}
+\dots\\
&=58x^2-\frac{8063}{14}x^4+O(x^6),
\end{split}
\end{equation}
where for the last line we have used the series corresponding to the critical path (\ref{seriesFortheCriticalPath}).
Notice further that
\begin{equation}
\begin{split}
f_1(Q)
&=147\sqrt{2}y^2+
	\frac{(2377x^2-1400y^2)y^2}{\sqrt{2}}
	+\mbox{ higher order terms}\\
&=609\sqrt{2}x^2-\frac{144213}{7\sqrt{2}}x^4+O(x^6).
\end{split}
\end{equation}
Here we have used again the series corresponding to the critical path (\ref{seriesFortheCriticalPath}).

Consequently
\begin{equation}
\frac{f_1}{\rho^2}
=\frac{21}{\sqrt{2}}
	-\frac{4107}{28\sqrt{2}}x^2
	+O(x^4),
\end{equation}
and similarly, 
\begin{equation}
f_2(Q)
=-\frac{21}{\sqrt{2}}
	-\frac{423\sqrt{2}}{7}x^2+O(x^4).
\end{equation}
Therefore we obtain that
\begin{equation}
\frac{f_1}{\rho^2}+f_2
\to
-\frac{7491}{28\sqrt{2}}x^2+O(x^4)
\end{equation}
and
\begin{equation}
\label{fOverRho3}
\frac{f_1}{\rho^4}+\frac{f_2}{\rho^2}
=-\frac{7491}{1624\sqrt{2}}+O(x^2)
=-\frac{7491}{1624\sqrt{2}}+O(\rho^2).
\end{equation}
This analysis shows that
in the equal mass case,
when the orbit approaches the rectilinear central configuration along the critical path, the limit of 
$f_1/\rho^4+f_2/\rho^2$  is finite.

%%%%%%%%%%%%%%%%%%%
%%%%%%%%%%%%%%%%%%%
%%%%%%%%%%%%%%%%%%%
\section{Saari's homographic conjecture for non-negative energy, II}
\label{two}
%%%%%%%%%%%%%%%%%%%
%%%%%%%%%%%%%%%%%%%
%%%%%%%%%%%%%%%%%%%
The main goal of this section is to prove the following result.
\begin{theorem}
\label{propositionForNonNegativeEnergy}
In the Newtonian equal-mass case of the three-body problem, Saari's homographic conjecture is true for all non-negative values of the energy.
\end{theorem}

Recall that by equations (\ref{eqPhi}), (\ref{eqMotForI}), and (\ref{eqIntegralForI}), $I \to \infty$ when $t \to \infty$.
We will distinguish several cases. Let us start with a simple one. 

\begin{proposition}
\label{lemmaNonNegativeHVanishingC}
In the Newtonian equal-mass case of the three-body problem, Saari's homographic conjecture is true for non-negative values of the energy and zero angular momentum. 
\end{proposition}
\begin{proof}
Since the angular momentum vanishes, i.e. $C=0$, condition (\ref{theCondition2}) becomes
\begin{equation}
\label{theCondition2ForCeq0}
-\sqrt{I}
=\frac{B}{9}
%%%%%%%%%
	\left(
		\frac{f_1}{\rho^4}+\frac{f_2}{\rho^2}
	\right).
\end{equation}
By (\ref{fOverRho3}), the right hand side of (\ref{theCondition2ForCeq0}) tends to a finite value when the motion 
approaches the origin, while the left hand side goes to infinity. This completes the proof. 
\end{proof}

In the non-zero angular momentum case, $C \ne 0$,
the key term in (\ref{theCondition2}) is
\[
2\epsilon C
	\sqrt{\frac{B-C^2}{3}}
	\;\frac{1}{\rho}.
\]
Since the left hand side of (\ref{theCondition2}) is negative, the sign factor $\epsilon$ must be chosen such that
\[
\epsilon C = -|C|.
\] 
Then condition (\ref{theCondition2}) becomes
\begin{equation}
\label{theCondition3}
-\rho \sqrt{I}
=\left( \frac{B-C^2}{9} \right)
	\left(
		\frac{f_1}{\rho^4}+\frac{f_2}{\rho^2}
	\right)
	\rho
	-
	2 |C|
	\sqrt{\frac{B-C^2}{3}}.
\end{equation}
To prove Theorem \ref{propositionForNonNegativeEnergy} for $C \ne 0$, we have to analyse the behaviour of $\rho(Q)\sqrt{I(q)}$ with respect to the fictitious time variable $\tau$. We will start by analysing $\rho(Q)$ and then continue with $I(q)$.

In Appendix \ref{secRhoByE}, we show that $\rho(Q)$ is given by
\[
\rho^2=-(E_1E_2+E_2E_3+E_3E_1),
\]
where $E_\ell$ are
defined by
\begin{equation}
\label{defOfE}
E_\ell(Q) = E_{jk}(Q)=m_j m_k \left( \frac{1}{r_{jk}^{a+2}(Q)}-\frac{a\mu}{M} \right).
\end{equation}
Therefore
\[
\frac{d}{d\tau}\rho^2
=-\frac{dE_1}{d\tau}(E_2+E_3)
	-\frac{dE_2}{d\tau}(E_3+E_1)
	-\frac{dE_3}{d\tau}(E_1+E_2).
\]
Then, using equations (\ref{defOfE}) and (\ref{drij}), we obtain
\[
\begin{split}
\rho^2 \frac{d\rho}{d\tau}
=-&\frac{a+2}{2}m_1 m_2 m_3 \epsilon \kappa \Delta\
	\sum \frac{1}{r_{jk}^{a+4}}
		\left( \frac{1}{r_{k\ell}^{a+2}}-\frac{1}{r_{\ell j}^{a+2}}\right)\\
	&\times	\left[
			m_k m_\ell \left(\frac{1}{r_{k\ell}^{a+2}}-\frac{a\mu}{M}\right)
			+
			m_\ell m_j \left(\frac{1}{r_{\ell j}^{a+2}}-\frac{a\mu}{M}\right)
		\right].
\end{split}
\]
Note that the right hand side is a function of $Q$.

To determine the behaviour of $\rho$ near the origin, we take $q_1=(-1,0)$, $q_2=(1,0)$, and $q_3=(x,y)$. Then we have
\[
\rho^2 \frac{d\rho}{d\tau}
=\frac{\sqrt{2}\epsilon \kappa}{3}
	x y
	(1218+7769x^2-6022y^2
	+\mbox{ higher order terms}).
\]

Consider now the orbit $x \to +0$ and $y \to +0$. Then, using the expression (\ref{seriesFortheCriticalPath}) for the critical path, we are led to
\[
\rho^2 \frac{d\rho}{d\tau}
=\epsilon \kappa 
	\left(
		58\sqrt{406}x^2+\frac{47576}{7}\sqrt{\frac{58}{7}}x^4+O(x^6)
	\right).
\]
The sign factor $\epsilon$ must be negative because we are considering the limit $\rho \to 0$. Then, using the expression (\ref{rho2SeriesFortheCriticalPath}) for $\rho^2$ near the origin along the critical path, we obtain
\[
\frac{d\rho}{d\tau}
=-\kappa\sqrt{406}
	\left(
		1
		-\frac{38711}{812\times406}\rho^2
		+O(\rho^4)
	\right).
\]
Thus, the asymptotic behaviour of $\rho$ is given by
\begin{equation}
\begin{split}
\label{behaviorOfRho}
\rho
&=\sqrt{406}\kappa (\tau_\infty - \tau)
	\left(
		1+O(\rho^2)
	\right)\\
&=\sqrt{406}\kappa (\tau_\infty - \tau)
	\left(
		1+O\left( (\tau_\infty - \tau)^2 \right)
	\right).
\end{split}
\end{equation}

We can now analyze the asymptotic behaviour of $I(q)$. From equations (\ref{eqPhi}), (\ref{eqMotForI}), and (\ref{eqIntegralForI}),
solutions with $H\ge 0$ have the property that $I \to \infty$ when $t \to \infty$. Also,
\[
\frac{dI}{dt}
=\sqrt{8HI+8\mu I^{1/2}-4B}.
\]
Since $d\tau = dt/I$, we obtain
\[
d\tau
=\frac{dI}{I \sqrt{8HI+8\mu I^{1/2}-4B}}.
\]
Integrating this equality, we are led to the function $\tau(I)$.
Then $I \to \infty$ corresponds to $\tau \to \tau(\infty)$,
which turns out to be finite.
We write $\tau(I)=\tau$ and $\tau(\infty)=\tau_\infty$. Then
%%%%%%%%%
\[
\tau_\infty - \tau
=\int_{\tau}^{\tau_\infty}d\tau
=\int_{I}^{\infty}\frac{dI}{I \sqrt{8HI+8\mu I^{1/2}-4B}}.
\]
We further split our discussion into two cases: $H=0$ and $H>0$.

For $H=0$,
\[
\tau_\infty - \tau
= \frac{1}{\sqrt{8\mu}}
	\int_{I}^{\infty} \frac{dI}{I^{1+1/4}}
	\left(
		1+O\left(\frac{1}{\sqrt{I}}\right)
	\right)
=\frac{1}{\sqrt{8\mu}} \frac{4}{I^{1/4}}
	\left(
		1+O\left(\frac{1}{\sqrt{I}}\right)
	\right).
\]
Thus, we obtain
\begin{equation}
\label{behaviorOfIForEeqZero}
\sqrt{I}
=
\frac{2}{\mu}\; \frac{1}{(\tau_\infty - \tau)^2}
	\left(
		1+O(\tau_\infty-\tau)^2
	\right).
\end{equation}
Combining equations (\ref{behaviorOfRho}) and
(\ref{behaviorOfIForEeqZero}), we have
\[
\rho(Q) \sqrt{I(q)}
=
\frac{2\sqrt{406}\kappa}{\mu}\; \frac{1}{(\tau_\infty-\tau)}
	\left(
		1+O(\tau_\infty-\tau)^2
	\right)
\to \infty.
\]
Therefore condition (\ref{theCondition3}) is not satisfied near the origin. We have thus proved the following result.
\begin{proposition}
\label{lemmaVanishingHNonVanishingC}
In the Newtonian equal-mass case of the three-body problem, Saari's homographic conjecture is true for zero energy and all non-zero angular momenta. 
\end{proposition}

For $H>0$,
\[
\begin{split}
\tau_\infty - \tau
&= \frac{1}{\sqrt{8H}}
	\int_{I}^{\infty} \frac{dI}{I^{1+1/2}}
	\left(
		1-\frac{\mu}{2H}\frac{1}{\sqrt{I}}+O\left(\frac{1}{I}\right)
	\right)\\
&=\frac{1}{\sqrt{8H}} 
	\left(
		\frac{2}{I^{1/2}}
		-\frac{\mu}{2H}\frac{1}{I}
		+O\left(\frac{1}{I^{3/2}}\right)
	\right).
\end{split}
\]
We thus obtain
\[
\begin{split}
\frac{1}{\sqrt{I}}
&=\sqrt{2H}(\tau_\infty - \tau)+\frac{\mu}{4H}\frac{1}{I}+O\left(\frac{1}{I^{3/2}}\right)\\
&=\sqrt{2H}(\tau_\infty - \tau)+\frac{\mu}{2}(\tau_\infty-\tau)^2+O\left((\tau_\infty-\tau)^3\right)\\
&=\sqrt{2H}(\tau_\infty - \tau)
	\left(
		1+\frac{\mu}{2\sqrt{2H}}(\tau_\infty-\tau)+O\left( (\tau_\infty-\tau)^2\right)
	\right).
\end{split}
\]
Therefore
\begin{equation}
\label{behaviorOfIForEPositive}
\sqrt{I}
=\frac{1}{\sqrt{2H}}\; \frac{1}{(\tau_\infty - \tau)}
	\left(
		1-\frac{\mu}{2\sqrt{2H}}(\tau_\infty-\tau)+O\left( (\tau_\infty-\tau)^2\right)
	\right).
\end{equation}
Then, for the left hand side of condition (\ref{theCondition3}), we have
\[
\begin{split}
-\rho \sqrt{I}
&=-\frac{\sqrt{406}}{\sqrt{2H}}
	\kappa
	\left(
		1
		- \frac{\mu}{2\sqrt{2H}} (\tau_\infty-\tau)
		+O(\tau_\infty-\tau)^2
	\right)\\
&=-\sqrt{\frac{406}{2H}}
	\kappa
	+\frac{\mu}{4H}\rho
	+O(\rho^2).
\end{split}
\]
Note that the coefficient of the linear term in $\rho$
is $\mu/4H$, which is positive.
Using equation (\ref{fOverRho3}), the right hand side of condition (\ref{theCondition3}) becomes
\[
\frac{B-C^2}{9}
	\left(
		-\frac{7491}{1624\sqrt{2}}
		+O(\rho^2)
	\right)
	\rho
-2|C| \sqrt{\frac{B-C^2}{3}},
\]
which means that the corresponding coefficient of the $\rho$ term is
negative. This implies that condition (\ref{theCondition3}) cannot be satisfied. Thus, we have proved the following result.
\begin{proposition}
\label{lemmaPositiveHNonVanishingC}
In the Newtonian equal-mass case of the three-body problem, Saari's homographic conjecture is true for positive energy and all non-zero angular momenta. 
\end{proposition}

The proof of Theorem \ref{propositionForNonNegativeEnergy} follows now from Propositions \ref{lemmaNonNegativeHVanishingC},
\ref{lemmaVanishingHNonVanishingC}, and \ref{lemmaPositiveHNonVanishingC}.

\medskip

We will end this section with the following result, inspired by the behaviour observed in equation (\ref{fOverRho3}).

\begin{proposition}
In the Newtonian equal-mass case of the three-body problem of constant configurational measure and non-zero angular momentum, if the particle system forms a central configuration at some initial time instant, then it maintains the same central configuration for all times.
\end{proposition}

\begin{proof}
This property is obvious for the equilateral central configurations because the equilateral shape is an isolated minimum of the configurational measure $\mu(Q)$.
For the rectilinear central configurations, 
if a solution escapes from a rectilinear central configuration
after having reached it at some time instant,
then the time-reversed orbit approaches this rectilinear configuration
along the critical path with finite $I(q)$.
Then, in condition (\ref{theCondition2}), the left hand side is finite, while the right hand side diverges as $1/\rho \to \infty$. This completes the proof.
\end{proof}

In principle, this result could be extended to the case $C=0$ by estimating the time dependence near the critical points on both sides of the equation
\begin{equation}
-\sqrt{I}
=\frac{B}{9}
%%%%%%%%%
	\left(
	\frac{f_1}{\rho^4}+\frac{f_2}{\rho^2}
	\right).
\end{equation}
%%%%%%%%

%%%%%%%%%%%%%%%%%%%%%%%%%%%%
%%%%%%%%%%%%%%%%%%%%%%%%%%%%
%%%%%%%%%%%%%%%%%%%%%%%%%%%%
\section{General results}
\label{final}
%%%%%%%%%%%%%%%%%%%%%%%%%%%%
%%%%%%%%%%%%%%%%%%%%%%%%%%%%
%%%%%%%%%%%%%%%%%%%%%%%%%%%%
In this final section, we will prove a few general results that clarify Saari's homographic conjecture for a large class of initial conditions.
As shown in Section \ref{kineticSection},
to prove Saari's homographic conjecture for the planar three-body problem
suffices to show that there are no solutions in case (iv),
namely when $\mu=$ constant, $B-C^2>0$, and $\rho\ne 0$.
In Sections \ref{one} and \ref{two}, we investigated 
the orbit in case (iv) for non-negative energy. 
In Section \ref{one}, we pointed out that we can prove Saari's homographic conjecture for non-negative energy
by showing that condition (\ref{theCondition2}) is violated
near the rectilinear central configurations along the critical path.
Then, in Section \ref{two}, we applied this idea for the Newtonian equal-mass case. For negative energy, we will further prove some additional properties.

Notice that, by equation (\ref{eqIntegralForI}), the moment of inertia oscillates for $H<0$ between $I_{\text{min}}$ and $I_{\text{max}}$, which are the two solutions of the equation $\Phi(I)=-2B$.
Then the fictitious time defined by equation (\ref{defOfTau}),
\[
\tau(t)
=\int_0^{\tau(t)} d\tau
=\int_0^t \frac{dt}{I(q)},
\]
is such that $\tau \to \infty$ when $t \to \infty$.

We can now prove the following lemmas.
\begin{lemma}
\label{distancezero}
For any solution of the planar three-body problem, there is a fictitious time instant, $\tau_0$, such that 
$dr_{jk}(Q(\tau_0))/d\tau =0$ for some $j,k$,
where
$r_{jk}(Q)=|Q_j-Q_k|=|q_j-q_k|/\sqrt{I(q)}$.
\end{lemma}
\begin{proof}
If the motion is homographic, then $r_{jk}(Q)=$ constant, therefore $dr_{jk}(Q)/d\tau=0$ for all $\tau$. Consider therefore a non-homographic solution. Since $I(Q)=1$, the only case in which $dr_{jk}/d\tau \ne 0$ for all $j, k$, and $\tau$ is the one in which $r_{jk}$ tends to a constant when $\tau \to \infty$. This happens when the coordinates $Q_k$ have a limiting shape: a triangle or a line.
%%%%%%%%%
Since the angular momentum is zero in Fujiwara coordinates,
the variables $Q_k$ approach limiting positions, so $P_k \to 0$.
Since $\sum |P_k|^2/m_k = B-C^2$ (constant), it follows that 
$\sum |P_k|^2/m_k =0$. 
Therefore $P_k=0, k=1,2,3$. 
Consequently the solution is homographic, a contradiction that proves the result.
\end{proof}

\begin{lemma}
\label{rectil-or-isos}
Every solution of the planar three-body problem with constant configurational measure passes through at least one rectilinear or one isosceles configuration.
\end{lemma}
\begin{proof}
By Lemma \ref{distancezero}, there is a time instant $\tau_0$
such that $dr_{jk}/d\tau=0$. From equation (\ref{drij}) in Appendix A,
the equation $dr_{jk}/d\tau=0$ implies that $\Delta=0$ (i.e. the area of the triangle vanishes) or that $r_{\ell j}=r_{k\ell}$. Therefore the corresponding configuration is either collinear or isosceles.
\end{proof}
According to Lemma \ref{rectil-or-isos}, we can prove the conjecture
for $H<0$ 
by showing that condition (\ref{theCondition2}) is violated
only near rectilinear and isosceles configurations.

\medskip

We can now add the following two results in case (iv) with no restrictions on the values of the energy constant.
%%%%%
\begin{theorem}
\label{momentum}
For a given configuration $q$ that is not a central configuration,
there is no constant configurational-measure solution of the planar three-body problem whose initial configuration is the given $q$ and for which $B-C^2>0$ is small enough.
\end{theorem}
\begin{proof}
Since the configuration $q$ is not central, necessarily $\rho \ne 0$.
Consider condition (\ref{theCondition2}) for a given configuration $q$.
Then $I(q)$, $f_1(Q)$, $f_2(Q)$, and $\rho(Q)$ are given. If we take $B-C^2$ small enough to depend on the given $q$, then the right hand side becomes very small while the left hand side is a given number. Under these circumstance, condition (\ref{theCondition2}) is not satisfied. This proves the result.
\end{proof}
\begin{theorem}
\label{equilateral}
For given values of $I(q)$, $C$, and $B-C^2>0$, there are no constant configurational-measure solutions of the planar three-body problem
whose initial configurations are close enough to an equilateral triangle.
\end{theorem}
 \begin{proof}
We will show that the right hand side in condition (\ref{theCondition2}) goes to infinity when $Q$ tends to an equilateral central configuration.
Condition (\ref{theCondition2}) will not be satisfied for this limit, a fact that will prove the result.
%%%%%%%%%%%%%%%%%%%%%%%%%%%%%%
 
To calculate the behaviour of $f_1(Q)$, $f_2(Q)$, and $\rho(Q)$ near the equilateral configurations, let us take
$q_1=(-1,0)$, $q_2=(1,0)$, and $q_3=(x, \sqrt{3}+y)$.
Then we have
\[
\begin{split}
f_1(Q)
=&\frac{3(a+2)^3 m_1 m_2 m_3}{64}
	\left(
		\frac{\sum m_jm_k}{M}
	\right)^2
	\left(
		\frac{\sum m_j m_k}{M}
	\right)^{3a/2}\\
	&
	\times\Bigg(
		m_1 (x+\sqrt{3}y)^2
		+
		m_2 (x-\sqrt{3}y)^2
		+
		4m_3 x^2
	\Bigg)
	+\dots
\end{split}
\]
and
\[
\begin{split}
\rho^2(Q)
=&\left(\frac{(a+2)m_1 m_2 m_3}{4M}\right)^2
	\left(
		\frac{\sum m_j m_k}{M}
	\right)^a\\
	&
	\times\Bigg(
		m_1^2 (x+\sqrt{3}y)^2
		+
		m_2^2 (x-\sqrt{3}y)^2
		+
		4m_3^2 x^2\\
		&-
		m_1 m_2 (x^2-3y^2)
		+
		2m_2 m_3 x(x-\sqrt{3}y)
		+
		2m_3 m_1 x(x+\sqrt{3}y)
	\Bigg)\\
	&+\mbox{ higher order terms}.
\end{split}
\]
It is easy to show that $f_2(Q) \to 0$ for $(x,y) \to (0,0)$. Therefore
\[
\frac{f_1(Q)}{\rho^2(Q)} + f_2(Q)
\to \mbox{ finite limit for }
(x,y) \to (0,0).
\]
Let us prove that in the general-mass case, the value of this limit is positive and depends on the direction of the path along which $(x,y) \to (0,0)$. Note that the dominant term in $\rho^2$,
\[
\begin{split}
&m_1^2 (x+\sqrt{3}y)^2
		+
		m_2^2 (x-\sqrt{3}y)^2
		+
		4m_3^2 x^2\\
		&-
		m_1 m_2 (x^2-3y^2)
		+
		2m_2 m_3 x(x-\sqrt{3}y)
		+
		2m_3 m_1 x(x+\sqrt{3}y),
\end{split}
\]
is positive definite for $(x,y)\ne(0,0)$.
To see this, it is convenient to write $\xi=x+\sqrt{3}y$ and $\eta=x-\sqrt{3}y$. Then the term appears as
\[
\begin{split}
&m_1^2 \xi^2+m_2^2 \eta^2+m_3^2 (\xi+\eta)^2\\
		&-m_1 m_2 \xi \eta+m_2 m_3 (\xi+\eta)\eta+m_3 m_1 (\xi+\eta)\xi\\
=&
	(m_1^2+m_3^2+m_3 m_1)\xi^2
	+(m_2^2+m_3^2+m_2 m_3)\eta^2\\
	&+(2m_3^2-m_1 m_2+m_2 m_3+m_3 m_1)\xi \eta.
\end{split}
\]
The discriminant is negative, as follows
\[
\begin{split}
&(2m_3^2-m_1 m_2+m_2 m_3+m_3 m_1)^2
-4(m_1^2+m_3^2+m_3 m_1)(m_2^2+m_3^2+m_2 m_3)\\
&=-3\left(\sum m_j m_k\right)^2
< 0.
\end{split}
\]
Therefore, the dominant term in the expansion of $\rho$ is positive definite in the second order for $x$ and $y$. It is obvious that $f_1$ is also positive definite in the same order for $x$ and $y$. Thus, $f_1/\rho^2$ tends to a positive finite value for $(x,y)\to (0,0)$.
But, for general masses, the limiting value depends on the direction of the approach. For example, if we approach the equilateral configuration vertically, for $x=0$ and $y \to 0$, then
\[
\frac{f_1}{\rho^2}
\to
\frac{3(a+2)(m_1+m_2)\left( \sum m_j m_k \right)^2}
	{4m_1 m_2 m_3 (m_1^2+m_1 m_2 +m_2^2)}
	\left(
		\frac{\sum m_j m_k}{M}
	\right)^{a/2}.
\]
On the other hand, for the limit $y=0$ and $x\to 0$, we get
\[
\begin{split}
\frac{f_1}{\rho^2}
\to
&\frac{3(a+2)(m_1+m_2+4m_3)
	\Big(\sum m_jm_k\Big)^2
	}{
	4m_1 m_2 m_3
	(m_1^2+m_2^2+4m_3^2-m_1 m_2+2m_2m_3+2m_3m_1)
	}
	\left(
		\frac{\sum m_j m_k}{M}
	\right)^{a/2}.
\end{split}
\]
For equal masses, the finite value does not depend on this direction, and
$f_1/\rho^2 + f_2 \to 9(a+2)/2$.
In any case,
\[
\frac{f_1(Q)}{\rho^4(Q)} + \frac{f_2(Q)}{\rho^2(Q)}
\sim \frac{\mbox{positive constant}}{\rho^2(Q)} \mbox{ for }
(x,y) \to (0,0).
\]
Therefore the right hand side of condition (\ref{theCondition2})
goes to infinity at the limit for any given value of $B-C^2>0$.
This completes the proof.
\end{proof}

\medskip

It is interesting to remark that, in a certain sense, Theorems \ref{momentum} and \ref{equilateral} complement each other. While Theorem \ref{momentum} shows that Saari's homographic conjecture is true for any initial positions if the initial velocities behave sufficiently well, Theorem \ref{equilateral} allows any initial velocities if the initial positions belong to some suitable region.

In the light of Lemma \ref{rectil-or-isos}, a complete proof of Saari's homographic conjecture can be given if condition (\ref{theCondition2}) is shown impossible when solutions come close to rectilinear and isosceles configurations. Only some technical aspects prevent us at this point from obtaining such a proof of this result. Nevertheless, our numerical evidence suggests that this approach is feasible.

\appendix

%%%%%%%%%%%%%%%%%%%
%%%%%%%%%%%%%%%%%%%
%%%%%%%%%%%%%%%%%%%
\section{Derivative of the mutual distances}
\label{secDrij}
%%%%%%%%%%%%%%%%%%%
%%%%%%%%%%%%%%%%%%%
%%%%%%%%%%%%%%%%%%%
In this appendix, we will provide a formula for the derivative of the mutual distances. 
Since $\sum m_\ell Q_\ell=0$, we have the identity
$$M m_\ell |Q_\ell|^2 + m_j m_k r_{jk}^2(Q) = m_j + m_k.$$
Differentiating it, we obtain
\begin{equation}
\label{drij0}
m_j m_k r_{jk}\frac{dr_{jk}}{d\tau}
=-MQ_\ell\cdot P_\ell
=\frac{\epsilon \kappa M}{\rho}
	Q_\ell \wedge G_\ell.
\end{equation}
A straightforward computation of $Q_\ell \wedge G_\ell$ yields
\begin{equation}
\label{torque}
Q_\ell \wedge G_\ell
=\frac{m_1m_2m_3\Delta}{M}
	\left(
		\frac{1}{r_{\ell j}^{a+2}}-\frac{1}{r_{k\ell}^{a+2}}
	\right)
\end{equation}
with $(j,k,\ell)=(1,2,3)$, $(2,3,1)$ or $(3,1,2)$, where $\Delta$ is twice the oriented area for the triangle $Q_1Q_2Q_3$,
$\Delta(Q)=Q_1\wedge Q_2+Q_2\wedge Q_3+Q_3\wedge Q_1$.
Note that
\[
m_1m_2Q_1\wedge Q_2
=m_2m_3Q_2\wedge Q_3
=m_3m_1Q_3\wedge Q_1
=\frac{m_1m_2m_3}{M}\Delta.
\]
Equations (\ref{drij0}) and (\ref{torque}) lead us to the desired formula,
\begin{equation}
\label{drij}
m_j m_k r_{jk}\frac{dr_{jk}}{d\tau}
=m_1 m_2 m_3 \frac{\epsilon \kappa \Delta}{\rho}
	\left(
		\frac{1}{r_{\ell j}^{a+2}}-\frac{1}{r_{k\ell}^{a+2}}
	\right).
\end{equation}

%%%%%%%%%%%%%%%%%%%
%%%%%%%%%%%%%%%%%%%
%%%%%%%%%%%%%%%%%%%
\section{Expression for $\rho$}
\label{secRhoByE}
%%%%%%%%%%%%%%%%%%%
%%%%%%%%%%%%%%%%%%%
%%%%%%%%%%%%%%%%%%%
%
In this appendix, we derive a useful expression for $\rho$.
The function $E_\ell(Q)=E_{jk}$ are defined by equation (\ref{defOfE}).
%%%%%%%%
Using the identity
$
\sum_{j \ne k}m_j m_k (Q_j-Q_k)=-Mm_k Q_k,
$
we have
\[
G_k(Q) = g_k(Q)+a\mu m_k Q_k=\sum_{j\ne k} (Q_j-Q_k)E_{jk}(Q).
\]
Note that
$$\sum r_{jk}^2(Q)E_\ell(Q) =0.$$
A straightforward computation yields the expression
\[
\left| G_\ell \right|^2 = -\left( \sum_{s < m} E_s E_m \right) r_{jk}^2.
\]
Comparing this expression with equation (\ref{SSTforQG}),
namely $|G_\ell|^2=\rho^2 r_{jk}^2$,
we obtain the desired formula
\begin{equation}
\label{rhoByE}
\rho^2 = - \sum_{s < m} E_s E_m.
\end{equation}

\noindent {A{\scriptsize CKNOWLEDGMENTS}.} We are indebted to the Banff International Research Station, which granted us a two-week Research in Teams residence during the period February 11-25, 2006. Although the collaboration that led to this paper began almost two years earlier, the intense brainstorming sessions that took place in Banff were crucial for obtaining the main results presented here. We would also like to thank Don Saari, Alain Chenciner, and Alain Albouy for reading an earlier version of this paper and suggesting improvements. Florin Diacu was also supported by an NSERC Discovery Grant and Ernesto P\'erez-Chavela by a CONACYT Grant.

\end{document}